\PassOptionsToClass{nonacm}{acmart}
\documentclass[sigconf,screen, natbib=false]{acmart}

\RequirePackage[
	datamodel=acmdatamodel,
	sortcites=true,
	style=acmnumeric
]{biblatex}
\usepackage{amsthm}
\usepackage{amsmath}
\usepackage{hyperref}
\usepackage[capitalize]{cleveref}
\usepackage[disable]{todonotes}
\usepackage{ifthen}
\usepackage[inline]{enumitem}
\usepackage{bm}

\usepackage{tikz}
\usepackage{tikz-qtree}

\usetikzlibrary{calc}

\DeclareUnicodeCharacter{0301}{*************************************}

\newcommand{\emptyWord}{\varepsilon}
\newcommand{\Lang}[1]{\mathsf{L}(#1)}

\newcommand{\Parikh}[1]{\Psi(#1)}

\newcommand{\yield}[1]{\mathsf{yield}(#1)}

\newcommand{\ParikhMap}{\Psi}

\newcommand{\powerset}[1]{\mathbb{P}(#1)}
\newcommand{\N}{\mathbb{N}}

\newcommand{\Q}{\mathbb{Q}}
\newcommand{\Z}{\mathbb{Z}}

\newcommand{\starRightarrow}{%
  \mathrel{\vbox{\offinterlineskip\ialign{%
    \hfil##\hfil\cr
    $\scriptscriptstyle\mathop{*}$\cr
    $\Rightarrow$\cr
}}}}

\newcommand{\grammarstep}[1][YYY]{\ifthenelse{\equal{#1}{YYY}}{\Rightarrow}{\Rightarrow_{#1}}}
\newcommand{\grammarsteps}[1][YYY]{\ifthenelse{\equal{#1}{YYY}}{\starRightarrow}{\starRightarrow_{#1}}}

\newcommand{\bzero}{\bm{0}}
\newcommand{\bs}{\bm{s}}
\newcommand{\bp}{\bm{p}}
\newcommand{\bu}{\bm{u}}
\newcommand{\bv}{\bm{v}}
\newcommand{\bw}{\bm{w}}
\newcommand{\bx}{\bm{x}}
\newcommand{\by}{\bm{y}}
\newcommand{\bz}{\bm{z}}
\newcommand{\bb}{\bm{b}}
\newcommand{\bc}{\bm{c}}

\newlist{conditions}{enumerate}{1}
\setlist[conditions,1]{label=(\roman*)}
\crefname{conditionsi}{condition}{conditions}
\Crefname{conditionsi}{Condition}{Conditions}

\newtheorem{theorem}{Theorem}
\newtheorem{proposition}[theorem]{Proposition}
\newtheorem{lemma}[theorem]{Lemma}

\begin{document}
\title[Slice closures of indexed languages and word equations with counting constraints]{Slice closures of indexed languages and \\ word equations with counting constraints}

\newcommand{\mysubsection}[1]{\subsubsection*{#1}}
\author{Laura Ciobanu}
\orcid{0000-0002-9451-1471}
\affiliation{%
  \institution{Hariot-Watt University}
  \streetaddress{P.O. Box 1212}
  \city{Edinburgh}
  \country{UK}
}
\email{L.Ciobanu@hw.ac.uk}

\author{Georg Zetzsche}
\orcid{0000-0002-6421-4388}
\affiliation{%
  \institution{Max Planck Institute for Software Systems (MPI-SWS)}
  \country{Germany}}
\email{georg@mpi-sws.org}

\begin{abstract}
Indexed languages are a classical notion in formal language theory. As the language equivalent of second-order pushdown automata, they have received considerable attention in higher-order model checking.
Unfortunately, counting properties are notoriously difficult to decide for indexed languages: So far, all results about non-regular counting properties show undecidability.

In this paper, we initiate the study of slice closures of (Parikh
images of) indexed languages. A slice is a set of vectors of natural
numbers such that membership of $\bu,\bu+\bv,\bu+\bw$ implies
membership of $\bu+\bv+\bw$. Our main result is that given an indexed language $L$, one can compute a semilinear representation of the smallest slice containing $L$'s Parikh image.

We present two applications. First, one can compute the set of all affine relations satisfied by the Parikh image of an indexed language. In particular, this answers affirmatively a question by Kobayashi: Is it decidable whether in a given indexed language, every word has the same number of $a$'s as $b$'s.

	As a second application, we show decidability of (systems of) word equations with rational constraints and a class of counting constraints: These allow us to look for solutions where a counting function (defined by an automaton) is \emph{not} zero. For example, one can decide whether a word equation with rational constraints has a solution where the number of occurrences of $a$ differs between variables $X$ and $Y$.
\end{abstract}

\keywords{indexed language,
higher-order,
slice,
semilinear,
effective,
word equation,
counting constraints,
length equations,
decidability,
computability}

\maketitle

\section{Introduction}
Indexed languages, introduced in the late 1960s by
Aho~\cite{DBLP:journals/jacm/Aho68}, are a classical notion in formal language
theory. They significantly extend context-free languages, but retain
decidability of basic properties such as emptiness.  Until the beginning of the
1990s, they have been studied with respect to equivalence with second-order
pushdown automata and safe second-order recursion
schemes~\cite{Maslov1976,DBLP:journals/tcs/Damm82,DBLP:journals/iandc/DammG86}, pumping properties~\cite{hayashi1973derivation} and
basic algorithmic properties, such as decidable emptiness~\cite{Maslov1974},
complexity of emptiness~\cite{DBLP:journals/iandc/Engelfriet91}, and decidability of
infinity~\cite{DBLP:conf/stoc/Rounds70,hayashi1973derivation}. 

\subsubsection*{Higher-order model checking} Over the last two decades, indexed languages have seen renewed interest because of their role in \emph{higher-order model checking}: In this approach to verification, one models the behavior of a functional program using a \emph{higher-order recursion scheme} (HORS), for example using Kobayashi's translation~\cite{DBLP:conf/popl/Kobayashi09}. A HORS generates a (potentially infinite) tree, but may also have a semantics as a finite-word language.  The translation into HORS enables the application of a wealth of algorithmic tools available for HORS. Here, a landmark result is the decidability of monadic second-order logic for HORS by Ong~\cite{DBLP:conf/lics/Ong06} (with alternative proofs by Kobayashi \& Ong~\cite{DBLP:conf/lics/KobayashiO09}, Hague, Murawski, Ong, and Serre~\cite{DBLP:conf/lics/HagueMOS08}, and Salvati and Walukiewicz~\cite{DBLP:conf/icalp/SalvatiW11}). More recently, this toolbox was extended to computation of downward closures for word languages of HORS~\cite{DBLP:conf/icalp/Zetzsche15,DBLP:conf/popl/HagueKO16,DBLP:conf/lics/ClementePSW16}, tree languages of safe HORS~\cite{DBLP:journals/fuin/BarozziniCCP22}, and logics expressing boundedness properties~\cite{DBLP:journals/fuin/BarozziniCCP22}.

\subsubsection*{Counting properties} Despite this rich set of algorithmic
results, HORS are notoriously difficult to analyze with respect to
\emph{counting properties}, i.e.\ properties that require precise
counting (without an apriori bound) of letters.  So far, all results about
counting properties show undecidability. Moreover, all these undecidability
properties already hold for languages of second-order recursion schemes, which are exactly the indexed languages~\cite{DBLP:conf/fossacs/AehligMO05}. 

To mention a particularly basic property: it is undecidable whether
an indexed language $L\subseteq\{a,b\}^*$ contains a word with the same number
of $a$'s and $b$'s~\cite[Prop.~7]{DBLP:journals/corr/Zetzsche15}. 
In fact, even inclusion in the Dyck language (the set of all
well-bracketed expressions over pairs of parentheses), which is decidable for
most other kinds of infinite-state
systems~\cite{DBLP:conf/icalp/0001GMTZ23a,DBLP:journals/pacmpl/BaumannGMTZ23,DBLP:journals/acta/BerstelB02,DBLP:conf/fossacs/TozawaM07,DBLP:journals/ipl/ManethS18,DBLP:journals/ijfcs/LobelLS21},
is undecidable for indexed languages, as shown in~\cite[Thm.~18]{DLT2016Uezato}
and~\cite[Thm.~3]{DBLP:journals/tcs/Kobayashi19}. 

These properties are important for verification: The former undecidability
implies that liveness verification is undecidable for second-order asynchronous
programs~\cite{DBLP:journals/lmcs/MajumdarTZ22} and Dyck inclusion is used to
verify reference counting
implementations~\cite{DBLP:journals/pacmpl/BaumannGMTZ23}.

Given this state of
affairs, Kobayashi writes about future directions in his 2013 article on
higher-order model checking~\cite{DBLP:journals/jacm/Kobayashi13}:
\begin{quote}On the theoretical side, it would be a challenge to identify
nonregular properties (such as counting properties) for which recursion scheme
model checking is decidable.\end{quote} 
and points out the negative results published
in~\cite{DBLP:journals/tcs/Kobayashi19}. As a concrete question, Kobayashi
leaves open in~\cite[Remark~1]{DBLP:journals/tcs/Kobayashi19} whether given an
indexed language $L\subseteq\{a,b\}^*$, it is decidable whether 
\begin{equation} L\subseteq\{w\in\{a,b\}^* \mid |w|_a=|w|_b\}. \label{kobayashi-inclusion}\end{equation}
This is a natural property to verify: If we view $a$'s a ``request'' operations
and $b$'s as ``response'' operations, then it says that in each run of the
system, the number of responses matches that of  requests.

\subsubsection*{Slices} In this paper, we identify slices as a surprisingly
useful concept for analyzing indexed languages. Introduced by Eilenberg and
Sch\"{u}tzenberger in 1969~\cite{eilenberg1969rational}, a \emph{slice} is a
subset $S$ of $\N^k$ such that if the vectors $\bu$, $\bu+\bv$, and $\bu+\bw$
belong to $S$, so does $\bu+\bv+\bw$. Eilenberg and Sch\"{u}tzenberger proved
that every slice is a semilinear set~\cite[Prop.~7.3]{eilenberg1969rational}
(an alternative proof, partly inspired by Hirshfeld's proof of the case of
congruences~\cite{hirshfeld1994congruences}, can be found
in~\cite{DBLP:conf/icalp/GanardiMPSZ22}). Eilenberg and Sch\"{u}tzenberger
deduce from this their much more well-known result that every congruence on
$\N^k$ is semilinear. Recall that an equivalence relation $\equiv$ on $\N^k$ is
a \emph{congruence} if $\bu\equiv \bv$ and $\bu'\equiv \bv'$, then
$\bu+\bu'\equiv\bv+\bv'$. When viewed as a subset of $\N^k\times\N^k$, it is
easy to see that every congruence is a slice in $\N^k\times\N^k$. 

The semilinearity of congruences (for which alternative proofs are by
Redei~\cite{Redei1965}, Freyd~\cite{Freyd1968}, and
Hirshfeld~\cite{hirshfeld1994congruences}) has been used in general
decidability results for bisimulation in Petri
nets~\cite{DBLP:conf/stacs/Jancar94}, implies semilinearity of mutual
reachability relations in Petri
nets~\cite{DBLP:conf/concur/Leroux11,DBLP:journals/corr/abs-2210-09931}, and
has a well-studied connection to polynomial
ideals~\cite{mayr1982complexity,DBLP:journals/jsc/KoppenhagenM99,DBLP:journals/iandc/KoppenhagenM00,DBLP:journals/jsc/KoppenhagenM01}.
Although slices are apparently more general than congruences and share
the properties of semilinearity and the ascending chain
condition~\cite{eilenberg1969rational}, slices have received surprisingly
little attention: We are only aware of Eilenberg and
Sch\"{u}tzenberger~\cite{eilenberg1969rational}, results by Grabowski on
effectively semilinear slices~\cite{DBLP:conf/fct/Grabowski81}, and an
algorithm for reachability in bidirected pushdown 
VASS~\cite{DBLP:conf/icalp/GanardiMPSZ22}.

\newcommand{\Dop}{\mathcal{D}}
\newcommand{\Sop}{\mathcal{S}}
\subsection*{Contribution}
Our main result is Theorem \ref{slice-closures}, which shows that given an indexed language $L\subseteq\Sigma^*$, one can
effectively compute a semilinear representation of the smallest slice 
containing $\Parikh{L}\subseteq\N^k$ (with $k=|\Sigma|$), the Parikh image of
$L$. Note that since the intersection of any collection of slices is again a
slice, every set $K\subseteq\N^k$ has a smallest slice containing $K$. We
call this the \emph{slice closure} of $K$ and denote it $\Sop^\omega(K)$.
Hence, we show (in Section \ref{sec:slice_closure}) that for every indexed
language $L$, its slice closure $\Sop^\omega(\Parikh{L})$ is effectively semilinear. 

\begin{theorem}\label{slice-closures}
	The slice closure of each indexed language is effectively semilinear.
\end{theorem}

We thus obtain an effectively semilinear overapproximation of $\Parikh{L}$ that
is amenable to algorithmic analysis (recall that semilinear sets are precisely
those expressible in Presburger arithmetic). 

Moreover, this overapproximation preserves pertinent counting properties of
$L$, meaning that \cref{slice-closures} yields (to our knowledge) the first
(non-regular) decidable counting properties of indexed languages: 

\begin{corollary}\label{prop:Kobayashi}
If $L$ is an indexed language on $\{a, b\}$ given by an indexed grammar, it is decidable whether every word has the same number of $a$'s as $b$'s.
\end{corollary}
To prove this, the key fact is that the vectors in
$\Parikh{L}$ all satisfy a system of linear equations over the integers if and only if
the vectors in $\Sop^\omega(\Parikh{L})$ do. One implication is clear. For the other direction, suppose that all vectors in $\Parikh{L}$ satisfy the system of linear equations $A\bx=\bb$ with $A\in\Z^{n\times k}$, $\bb\in\Z^n$. Since
the set $S$ of solutions to this system in $\N^k$ is a slice (see Example \ref{example:slices}) and $\Parikh{L} \subseteq S$, we must have $\Sop^\omega(\Parikh{L}) \subseteq S$ by the minimality of $\Sop^\omega(\Parikh{L})$; thus the vectors in $\Sop^\omega(\Parikh{L})$ satisfy the system $A\bx=\bb$.  Therefore, in particular, we give a positive answer to
Kobayashi's question and prove \cref{prop:Kobayashi}: One can decide the inclusion \eqref{kobayashi-inclusion}
by checking whether $\Sop^\omega(\Parikh{L})\subseteq\{(m,m) \mid m\in\N\}$.

\newcommand{\aff}{\mathsf{aff}}
Combined with a result by Karr~\cite{DBLP:journals/acta/Karr76}
(see~\cite[Theorem~1]{DBLP:conf/icalp/Muller-OlmS04} for a simpler and more
explicit proof), \cref{slice-closures} also implies that we can compute the  ``strongest affine invariant'' of
$\Parikh{L}$. 
An \emph{affine relation} is a vector $\bu\in\Z^k$ and a number $c\in\Z$,
and it is \emph{satisfied by $K\subseteq\N^k$} if $\bu^\top\bx=c$ for every
$\bx\in K$.
The (rational) \emph{affine hull} of a set $U\subseteq\Q^k$ is the set
\begin{multline*} \aff(U)=\{\lambda_1\bu_1+\cdots+\lambda_n\bu_n \mid n\ge 0,~\bu_1,\ldots,\bu_n\in U, \\
	\lambda_1,\ldots,\lambda_n\in\Q,~\lambda_1+\cdots+\lambda_n=1\}. \end{multline*}
A non-empty set $A\subseteq\Q^k$ is an \emph{affine space} if $\aff(A)=A$. It
is a basic fact in linear algebra that affine spaces are exactly those subsets
of the form $\bv+V$, where $\bv\in\Q^k$ and $V\subseteq\Q^k$ is a linear
subspace. Equivalently, affine spaces are exactly the non-empty solution sets
in $\Q^k$ to equation systems $B\bx=\bc$ with $B\in\Q^{m\times k}$ and
$\bc\in\Q^m$. By \emph{computing $\aff(U)$} for some given set $U\subseteq\Q^k$
we mean computing some $B\in\Q^{m\times k}$ and $\bc\in\Q^m$ such that
$\aff(U)=\{\bx\in\Q^k \mid B\bx=\bc\}$. Some authors call this equation system the
\emph{strongest affine invariant} of $U$, meaning a finite set of affine relations satisfied by
$\Parikh{L}$ that implies every affine relation satisfied by $\Parikh{L}$.
\begin{corollary}\label{affine-hull}
	Given an indexed language $L$, we can effectively compute $\aff(\Parikh{L})$.
\end{corollary}
We derive \cref{affine-hull} from \cref{slice-closures} in \cref{basic-notions}.

\subsection*{Word equations with counting constraints}
We also present an application of \cref{slice-closures} to word equations.
Let $A$ and $\Omega$ be finite alphabets, where we consider $\Omega=\{X_1, \dots, X_n\}$ as a collection of
variables and $A$ as letters. A \emph{word equation} over $(A,\Omega)$ is a
pair $(U,V)$ with $U,V\in (A\cup\Omega)^*$, where $U$ and $V$ should be seen as the two sides
of an equation. The core question is then whether the equation has any solutions, that is, whether we can substitute the variable occurrences in $U$
and $V$ by words over $A$ so that the two sides become equal. 
More formally, a solution to $(U,V)$ is a morphism $\sigma:(A\cup\Omega)^* \to A^*$ that fixes each $a\in A$ such that $\sigma(U)=\sigma(V)$. 
\begin{example}
Let $A=\{a, b\}$ and $\Omega=\{X, Y, Z\}$. The pair of words $(X^2aY,YZ^2a)$ represents the word equation $X^2aY=YZ^2a$. One possible solution is $\sigma(X)=ab, \sigma(Y)=a, \sigma(Z)=ba$.
\end{example}

\newcommand{\PSPACE}{\mathsf{PSPACE}}
The study of word equations has a long history. Decidability (of solvability) was shown by Makanin~\cite{Makanin1977} in 1977, and a $\PSPACE$ upper bound by Plandowski~\cite{Plandowski}. This was extended to word equations with \emph{rational constraints} by Diekert, Guti\'{e}rrez, and Hagenah~\cite{dgh05IC}. 

However, imposing \emph{counting constraints}, e.g.\ that the number of $a$'s in $X_1$ is the same as in $X_2$, makes word equations much more difficult. In a parallel to indexed languages, allowing such conditions without restrictions (we make this more precise in \cref{sec:wordeqs}) leads to undecidability, as shown by B\"{u}chi and Senger~\cite{buchi1990definability}. Whether word equations with rational constraints and \emph{length equations}, which asks e.g.\ that $X_1$ have the same length as $X_2$, are decidable remains a major open problem~\cite{WE_with_length_constraints_survey}, even in the case of quadratic equations~\cite{majumdar2021quadratic}. More generally, despite a lot of recent interest in extensions of word equations, most results show undecidability. Moreover, there is no decidability result that extends general word equations with rational constraints and any kind of counting constraints. See~\cite{Amadini,majumdar2021quadratic,DayManeaWE} for recent surveys.

Our application of \cref{slice-closures} introduces a class of counting constraints---we call them \emph{counting inequations}---for which word equations (with rational constraints) remain decidable:
\begin{theorem}\label{counting-inequations-decidable}
The problem of word equations with rational constraints and counting inequations is decidable.
\end{theorem}
Here, we rely on a recent result by Ciobanu, Diekert, and
Elder~\cite{DBLP:conf/icalp/CiobanuDE15} saying that solution sets of word
equations (with rational constraints) are EDT0L, and thus indexed languages.
For example, it is decidable whether there is a solution $\sigma$ to a word
equation $(U,V)$ such that $|\sigma(X_1)|_a\ne|\sigma(X_2)|_a$. To our
knowledge, this is the first strict extension of word equations with rational
constraints by any kind of counting constraints. We prove
\cref{counting-inequations-decidable} in \cref{sec:wordeqs}.

\subsection*{Key ingredients}
Our algorithm for \cref{slice-closures} is a \emph{saturation procedure} where for some monotone operator $F$, we compute a fixpoint of $F$ by applying $F$ repeatedly. By relying on an ascending chain condition, such algorithms terminate with a fixpoint after finitely many steps.

Saturation is a ubiquitous idea, but closest to our approach are algorithms where the ascending chain condition comes from (i)~dimension of vector spaces or affine spaces, (ii)~Hilbert's Basis Theorem, which implies that chains of polynomial ideals become stationary, (iii)~well-quasi-orderings, where chains of upward closed sets become stationary. Examples of (i) are Eilenberg's algorithm for equivalence of weighted automata over fields~\cite[p.143--145]{Eilenberg1974} and Karr's computation of strongest affine invariants~\cite{DBLP:journals/acta/Karr76} (see also~\cite{DBLP:conf/icalp/Muller-OlmS04} and the stronger result~\cite{DBLP:journals/jacm/HrushovskiOPW23}) and applications~\cite{DBLP:journals/tcs/HarjuK91}. Examples for (ii) are the algorithm for polynomial invariants of affine programs~\cite{DBLP:journals/ipl/Muller-OlmS04}, several algorithms for equivalence of transducers or register machines~\cite{DBLP:conf/lics/BenediktDSW17,DBLP:conf/focs/SeidlMK15,DBLP:conf/rp/FiliotR17,DBLP:conf/fsttcs/BoiretPS18} (see~\cite{DBLP:journals/siglog/Bojanczyk19a} for a unifying result). An example of (iii) is the classical backward search algorithm to decide coverability in well-structured transition systems~\cite{DBLP:journals/tcs/FinkelS01,DBLP:conf/lics/AbdullaCJT96}. Saturation over finite domains---also an example of (iii)---is used prominently in the analysis of pushdown systems~\cite{DBLP:conf/concur/BouajjaniEM97}.
(Some algorithms do not compute an ascending chain, but rather guess a fixpoint of $F$, e.g.~\cite{DBLP:journals/pacmpl/ChistikovMS22,DBLP:conf/icalp/GanardiMPSZ22,DBLP:journals/siglog/Bojanczyk19a,DBLP:journals/iandc/BlondinFM18}).

\subsubsection*{Slices} The key insight of this work is that the slice closure
of a set preserves enough information to permit a saturation procedure for
indexed languages that preserves counting properties: Our procedure crucially
relies on properties of slices that the concepts (i--iii) do not have. For
(iii), this is obvious, because the upward closure of indexed languages does
not preserve any precise counting properties\footnote{Upward closures for
indexed languages are known to be computable using this approach (see
e.g.~\cite[p.124]{DBLP:phd/dnb/Zetzsche16}).}. For (i) and (ii), we explain
this now. 

\subsubsection*{A fixpoint characterization}
On the technical level, one ingredient is a fixpoint characterization
(\cref{fixpoint-indexed}) of Parikh images of sentential forms---that consist
just of non-terminals and terminals (and no indexes)---of indexed
grammars\footnote{In fact, the proof could easily be modified to yield a
fixpoint characterization of word languages just in terms of sentential forms
without indexes.}. This appears to be new: There are fixpoint characterizations
of indexed languages, but they concern sentential forms \emph{including
indexes}~\cite{DBLP:conf/sas/CampionPG19}. Here, the absence of indexes is key
to having a fixpoint operator just on vectors of natural numbers.

This fixpoint characterization---via an operator called $\Dop$ (for
derivation)---allows us to exploit slices: Our algorithm iteratively applies
$\Dop$ to obtain a larger set of vectors, and then immediately computes the
slice closure each time. This is possible because each intermediate set is
semilinear and it is known that slice closures of semilinear sets are
computable (see \cref{slice-closure-semilinear}). Since slices satisfy an
ascending chain condition, this will result in the slice closure of the
Parikh images of sentential forms (without indexes).

However, a key property of slices is required in the end: We need to extract
the slice closure of the \emph{generated language}, i.e.\ those sentential forms that
\emph{only contain terminals}, from the slice closure of \emph{all} sentential forms.
We show that this is possible for slices, but that such an extraction is not
possible when using the methods (i) and (ii) above: For example, even to just
compute the affine hull of $\Parikh{L}$, it would
not suffice to compute the affine hull (or even the Zariski closure!) of the Parikh images of sentential forms. We prove this in the discussion after \cref{intersection-commutes}.

\subsection*{Organization}
The paper is organized as follows. In \cref{basic-notions}, we fix notations
and introduce slices, indexed grammars, and other notions. In
\cref{sec:slice_closure}, we prove our main result \cref{slice-closures}. In
\cref{sec:wordeqs}, we present and prove \cref{counting-inequations-decidable}.

\section{Basic notions: Slices and indexed grammars}\label{basic-notions}

\mysubsection{Notation}
For a finite alphabet $\Sigma$, we write $\N^\Sigma$ for the set of maps
$\Sigma\to\N$. We will consider $\N^\Sigma$ as a commutative monoid with
pointwise addition. If $|\Sigma|=k$, then $\N^\Sigma$ is clearly isomorphic to
$\N^k$ and so we will often transfer notation and terminology from one to the
other (in case the chosen isomorphism doesn't matter).  Moreover, we identify
each element $a\in\Sigma$ with the element $\bu_a\in\N^\Sigma$ with
$\bu_a(a)=1$ and $\bu_a(b)=0$ for all $b\in\Sigma\setminus\{a\}$.  For each
finite alphabet $\Sigma$, the \emph{Parikh map}
$\ParikhMap\colon\Sigma^*\to\N^\Sigma$ is the map with $\Parikh{w}(a)=|w|_a$
for $w\in\Sigma^*$ and $a\in\Sigma$, where $|w|_a$ is the number of occurrences
of $a$ in the word $w$.
For a set $S$, we denote its powerset by $\powerset{S}$.
\mysubsection{Semilinear sets}
A subset $K\subseteq\N^k$ is \emph{linear} if there is a vector $\bb\in\N^k$ and a finite subset $P=\{\bp_1,\ldots,\bp_r\}\subseteq\N^k$ such that $K=\bb+\langle P\rangle$, where we use the notation
\[ \bb+\langle P\rangle =\{\bb+\lambda_1\bp_1+\cdots+\lambda_r\bp_r \mid \lambda_1,\ldots,\lambda_r\in\N\}. \]
A subset $K\subseteq\N^k$ is \emph{semilinear} if it is a finite union of linear sets.
In particular, a semilinear set has a \emph{semilinear representation} $R=(\bb_1,P_1;\ldots,\bb_n,P_n)$ consisting of vectors $\bb_1,\ldots,\bb_n\in\N^k$ and finite subsets $P_1,\ldots,P_n\subseteq\N^k$. The set \emph{represented by $R$}, denoted $\langle R\rangle$, is then $\bb_1+\langle P_1\rangle \cup \cdots\cup \bb_n+\langle P_n\rangle$.

In this work, we distinguish between (i)~sets for which we
can prove semilinearity (potentially non-constructively) and (ii)~sets for
which we can even \emph{compute} a semilinear representation algorithmically,
from the available data. In the latter case, we speak of \emph{effectively
semilinear sets}. When we use this terminology, it will always be clear from
the context what the available data is (e.g.\ an indexed grammar, a semilinear
representation of a related set, etc.).

\emph{Presburger arithmetic} is the first-order theory of $(\N,+)$, the natural numbers with addition. It is well-known that a subset of $\N^k$ is definable in Presburger arithmetic if and only if it is semilinear~\cite[Theorem~1.3]{ginsburg1966semigroups}. This equivalence is effective, meaning we can algorithmically translate between Presburger formulas and semilinear representations.

\mysubsection{Slices and slice closures}
A subset $K\subseteq\N^k$ is a \emph{slice} if for any $\bu,\bu+\bv,\bu+\bw\in
K$ with $\bu,\bv,\bw\in\N^k$, we also have $\bu+\bv+\bw\in K$. We also say that
$\bu+\bv+\bw$ is \emph{obtained from $\bu,\bu+\bv,\bu+\bw$ using the slice
operation}. Clearly, the intersection of any collection of slices is also a slice. Thus, for every subset $K\subseteq\N^k$, there is a smallest slice that includes $K$. We call it the \emph{slice closure of $K$} and denote it by $\Sop^\omega(K)$.

\begin{example}\label{example:slices} The following are three kinds of examples of slices. 
\begin{enumerate}[label=(\roman*)]
	\item\label{example-solution-sets} The first is \emph{solution sets to linear diophantine equations}: For a matrix $A\in\Z^{n\times k}$ and a vector $\bb\in\N^n$, the solution set $\{\bx\in\N^k \mid A\bx=\bb\}$ is a slice: If $\bu,\bu+\bv,\bu+\bw$ are solutions, then clearly $A\bv=\bzero$ and $A\bw=\bzero$ and thus $A(\bu+\bv+\bw)=\bb$.

\item Another example is \emph{upward closed sets}: We say that $K\subseteq\N^k$ is \emph{upward closed} if $\bu\in K$ and $\bu\le\bv$ imply $\bv\in K$, where $\leq$ is component-wise inequality. Clearly, every upward closed set is a slice. 

\item Third, subtractive monoids are slices: we say that $K\subseteq\N^k$ is a \emph{subtractive submonoid} if $\bu,\bu+\bv\in K$ imply that $\bv\in K$.
\end{enumerate}
\end{example}

We will rely on two facts about slices by Eilenberg and Sch\"{u}tzenberger. The first is that every slice is semilinear:
\begin{theorem}[Eilenberg \& Sch\"{u}tzenberger]\label{slices-semilinear}
	Every slice is semilinear.
\end{theorem}
This is shown in \cite[Proposition 7.3]{eilenberg1969rational} (see
also~\cite{DBLP:conf/icalp/GanardiMPSZ22} for an alternative proof). 
Note that whether a given slice is \emph{effectively} semilinear depends on how it is represented.
In fact, there are many situations in which one cannot compute a semilinear representation for a given slice. For example, the set of reachable configurations in a lossy counter system is always a slice (since it is downward closed), but it cannot be computed for a given lossy counter system~\cite[Remark 14]{DBLP:journals/tcs/Mayr03}.
We will also rely on the ascending chain condition:
\begin{proposition}[Eilenberg Sch\"{u}tzenberger]
	If $S_1\subseteq S_2\subseteq\cdots\subseteq\N^k$ are slices, then
	there is an index $m$ with $S_m=\bigcup_{i\ge 0} S_i$.
\end{proposition}
This was also shown by Eilenberg and Sch\"{u}tzenberger~\cite[Corollary
12.3]{eilenberg1969rational} and since it follows easily from
\cref{slices-semilinear}, we include the short proof. 
\begin{proof}
	Since $S:=\bigcup_{i\ge 0} S_i$ is a slice and thus semilinear, we can
	write $S=\bb_1+\langle P_1\rangle~\cup~\cdots~\cup~\bb_n+\langle
	P_n\rangle$ for some semilinear representation
	$(\bb_1,P_1;\cdots;\bb_n,P_n)$. Consider the finite set
	$F=\{\bb_1\}\cup \bb_1+P_1~\cup~\cdots~\cup~\{\bb_n\}\cup\bb_n+P_n$.
	Since the entire set $\bb_i+\langle P_i\rangle$ can be built from
	$\{\bb_i\}\cup(\bb_i+P_i)$ using slice operations, we know that in
	particular, all of $S$ can be obtained from $F$ using slice operations,
	thus $S\subseteq\Sop^\omega(F)$. Moreover, as a finite set, $F$ must be
	included in $S_m$ for some $m$. Thus, we have
	$S\subseteq\Sop^\omega(F)\subseteq S_m$ since $S_m$ is a slice.
\end{proof}

With these basic facts in hand, we are ready to derive \cref{affine-hull} from \cref{slice-closures}.
\begin{proof}[Proof of \cref{affine-hull}]
	We begin by observing, as in
\cref{example:slices}\labelcref{example-solution-sets}, that $A\cap\N^k$ is a slice for every affine
space $A\subseteq\Q^k$.

	Now we claim that for every set $U\subseteq\N^k$, we have
	$\aff(U)=\aff(\Sop^\omega(U))$. The inclusion
	$\aff(U)\subseteq\aff(\Sop^\omega(U))$ is trivial. For the other
	inclusion, we use that $\aff(U)$ is the intersection of all affine spaces
	in $\Q^k$ that include $U$, together with the following observation: For every affine space
	$A\subseteq\Q^k$ with $U\subseteq A$, we also have
	$\Sop^\omega(U)\subseteq A$, because $A\cap\N^k$ is a slice, and $A\cap\N^k$ must contain $\Sop^\omega(U)$ by the minimality of $\Sop^\omega(U)$ as slice. Hence,
	$\aff(U)=\aff(\Sop^\omega(U))$.

	By \cref{slice-closures}, we can compute a semilinear representation
	for $\Sop^\omega(\Parikh{L})$. Since the affine hull can be computed
	for affine programs (this was implicitly shown by
	Karr~\cite{DBLP:journals/acta/Karr76} and explicitly and simplified by
	M\"{u}ller-Olm and Seidl~\cite[Theorem
	1]{DBLP:conf/icalp/Muller-OlmS04}), it follows that one can
	compute the affine hull of a given semilinear set.
\end{proof}

\mysubsection{Indexed languages}
Let us define indexed grammars. The following definition is a slight
variation
of the one
from~\cite{HopcroftUllman1979}.  An \emph{indexed grammar} is a tuple
$G=(N,T,I,P,S)$, where $N$, $T$, and $I$ are pairwise disjoint alphabets,
called the \emph{nonterminals}, \emph{terminals}, and \emph{index symbols},
respectively.  $S$ is the starting non-terminal symbol of the grammar and $P$ is the finite set of \emph{productions} of the forms 
\begin{align}
	&\text{(i) $A\to Bf$},		&\text{(ii) $Af\to B$},		& &\text{(iii) $A\to w$}, \label{production-forms}
\end{align}
where $A,B\in N$, $f\in I$, and $w\in N^*\cup T^*$.
We regard a word $Af_1\cdots f_n$ with $f_1,\ldots,f_n\in I$ as a nonterminal
to which a stack is attached. Here, $f_1$ is the topmost symbol and $f_n$ is on
the bottom.  For $w\in (N\cup T)^*$ and $x\in I^*$, we denote by $[w,x]$ the
word obtained by replacing each $A\in N$ in $w$ by $Ax$. A word in $(NI^*\cup
T)^*$ is called a \emph{sentential form}. For $q,r\in (NI^*\cup T)^*$, we write
$q\grammarstep[G] r$ if there are words $q_1,q_2\in (NI^*\cup T)^*$, $A\in N$,
$p\in (N\cup T)^*$ and $x,y\in I^*$ such that $q=q_1 Ax q_2$, $r=q_1[p,y]q_2$,
and one of the following holds: 
\begin{enumerate}[label=(\roman*)]
	\item $A\to Bf$ is in $P$, $p=B$ and $y=fx$, or
	\item $Af\to B$ is in $P$, $p=B$ and $x=fy$, or 
	\item $A\to p$ is in $P$, $p\in N^*\cup T^*$, and $y=x$.
\end{enumerate}
When the grammar is clear from the context, we also write $\grammarstep$ instead of $\grammarstep[G]$.
The language
\emph{generated by $G$} is 
\[ \Lang{G}=\{ w\in T^* \mid S\grammarsteps w\}, \]
where $\grammarsteps$ denotes the reflexive transitive closure of
$\grammarstep$. \emph{Derivation trees} are always unranked trees with
labels in $NI^*\cup T\cup \{\varepsilon\}$ and a very straightforward analog to
those of context-free grammars (see, for example, \cite{Smith2014}).  If $t$ is
a labeled tree, then its \emph{yield}, denoted $\yield{t}$, is the word spelled
by the labels of its leaves.
Since the definition of derivation trees is somewhat cumbersome
but a very straightforward analog to those of context-free grammars, we refer the reader
to \cite{HopcroftUllman1979} for details.

The productions of the forms (i), (ii), and (iii) from \cref{production-forms} will be called \emph{push},
\emph{pop}, and \emph{context-free} productions, respectively. Moreover, a context-free production $A\to w$ with $w\in T^*$ is also called a \emph{terminal} production.

\begin{example}\label{examplegrammar}
	Consider the indexed grammar $G=(N,T,I,P,S)$ with $N=\{S,R,U,V,W,A,B\}$, $I=\{f,g,\bot\}$, and $T=\{a,b\}$, where $P$ consists of the productions
	\begin{align*}
		&S\to R\bot, & &R\to Rf, & &R\to Rg,   & &R\to UU, \\
		&Uf\to V,    & &V\to UA, & &Ug\to W,   & &W\to UB, \\
		&A\to a,     & &B\to b,  & &U\bot\to E, & &E\to\varepsilon.
	\end{align*}
	The productions in the first row allow us to derive $S\grammarsteps Ux\bot Ux\bot$ for every index word $x\in\{f,g\}$. Those in the second row can then derive $Ufx\bot\grammarsteps Ux\bot Ax\bot$ and $Ugx\bot\grammarsteps Ux\bot Bx\bot$. Using the third row production, we can then produce terminal words: $A$'s and $B$'s become $a$'s and $b$'s, whereas $U\bot$ can be turned into $\varepsilon$. Note that the role of $\bot$ is to prevent $Ux\bot$ from being rewritten to $\varepsilon$ before consuming all of $x$.

	In particular, for every $x\in\{f,g\}^*$, we can derive $Ux\bot \grammarsteps U\bot w$, where $w$ is the word obtained from $x$ by replacing $f$ with $a$ and $g$ with $b$. Thus, we can derive $S\grammarsteps Ux\bot Ux\bot\grammarsteps ww$ for any such $x\in\{f,g\}^*$ and $w\in\{a,b\}^*$. The generated language is 
	\[ \Lang{G}=\{ww \mid w\in\{a,b\}^*\}. \]
	An example derivation tree for $w=abab$ is shown in \cref{derivationtree}.
\end{example}

\begin{figure}
\begin{center}
\begin{tikzpicture}
\tikzset{level distance=25pt}
\Tree [ .{$S$}
        [ .{$R\bot$}
          [ .{$Rf\bot$}
            [ .{$Rgf\bot$}
              [ .{$Ugf\bot$}
	        [ .{$Wf\bot$}
                  [ .{$Uf\bot$}
                    [ .{$V\bot$}
                      [ .{$U\bot$}
                        [ .{$E$} 
			  [ .{$\varepsilon$} ] ] ]
                      [ .{$A\bot$} 
			[ .{$a$} ] ] ] ]
                  [ .{$Bf\bot$}
		    [ .{$b$} ] ] ] 
		    ]
              [ .{$Ugf\bot$}
	        [ .{$Wf\bot$}
                  [ .{$Uf\bot$}
                    [ .{$V\bot$}
                      [ .{$U\bot$}
                        [ .{$E$} 
			  [ .{$\varepsilon$} ] ] ]
                      [ .{$A\bot$} 
			[ .{$a$} ] ] ] ]
                  [ .{$Bf\bot$}
		    [ .{$b$} ] ] ] 
		  ] ] ] ] ]
\end{tikzpicture}
\end{center}
\caption{Derivation tree for the grammar in \Cref{examplegrammar} with yield $abab$.}\label{derivationtree}
\end{figure}
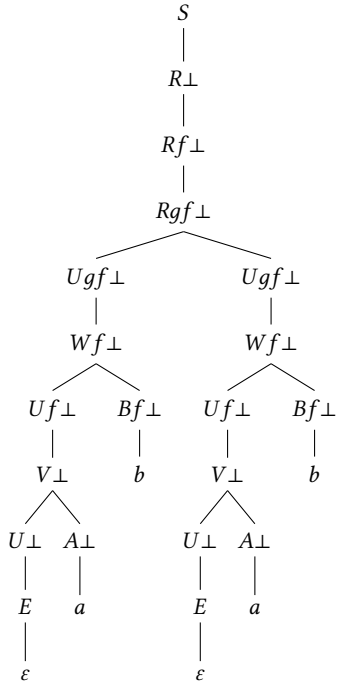

\section{Slice closures of indexed languages}\label{sec:slice_closure}

\newcommand{\tagPlug}{$\Dop$-1}
\newcommand{\tagStack}{$\Dop$-2}
\newcommand{\iniTuple}{\mathcal{I}}
In this section, we prove \cref{slice-closures}. To this end, fix an indexed grammar $G=(N,T,I,P,S)$.

A key ingredient in our proof is a fixpoint characterization of the Parikh image of indexed languages.
We describe it in terms of tuples of vector sets: One set for each non-terminal of the indexed grammar.
Formally, a \emph{set tuple} consists of a subset of $\N^{N\cup T}$ for each $A\in N$. In other words, a set tuple is an element of $\powerset{\N^{N\cup T}}^N$. 
For a set tuple $K\in \powerset{\N^{N\cup T}}^N$, we also write $K=(K_A)_{A\in N}$, where each $K_A\subseteq \N^{N\cup T}$. 

To the grammar $G$, we associate the set tuple $\Delta(G)\in\powerset{\N^{N\cup T}}^N$, which contains the Parikh images of all words over $N\cup T$ that can be derived from a given non-terminal. Thus, for each $A\in N$:
\[ \Delta(G)_A = \{\Parikh{w} \mid w\in (N\cup T)^*,~A\grammarsteps w \}. \]
\subsection{Indexed languages as fixpoints}
\newcommand{\grammarstepi}[1]{\grammarstep[\mathsf{pop}(#1)]}
\newcommand{\grammarstepcf}{\grammarstep[\mathsf{(iii)}]}
An \emph{operator} is a map that turns set tuples into set tuples, in other words a map $F:\powerset{\N^{N\cup T}}^N\to \powerset{\N^{N\cup T}}^N$ with $K\subseteq F(K)$ for each set tuple $K$. For such an operator $F$ and a set tuple $K$, we define
\begin{align} F^0(K)&=K & F^{i+1}(K)&=F(F^i(K)), &  F^\omega(K)&=\bigcup_{j\ge 0} F^j(K), \label{operator-powers}\end{align}
	for $i\ge 0$. Since $F^\omega(K)$ is a fixpoint of $F$, we will refer to descriptions via $F^\omega(K)$ for some $F$ and $K$ as \emph{fixpoint characterizations}.

\subsubsection*{Decomposing derivations in indexed grammars}

We now define an operator $\Dop$ that will permit such a fixpoint characterization of $\Delta(G)$ (see \cref{fixpoint-indexed}). 

The idea is the following. Consider derivation trees in indexed
grammars that produce a sentential form in $(N\cup T)^*$. These can be
transformed into new derivation trees in two ways. First, we can ``plug
in'' one tree with root non-terminal $A$ into another tree where $A$
occurs at a leaf. We call this rule ``\tagPlug''. Second, we
can take a derivation tree $t$ rooted in $B$ and extend it both at the
top and the leaves: At the top, we add a production $A\to Bf$ that
pushes $f$, and insert $f$ at the stack bottom of every node in $t$.
Then, at every leaf labeled with a non-terminal $C$, we apply a
pop production $Cf\to D$ to get rid of all the new $f$'s. As long as all
the non-terminal leaves of $t$ allow such pop productions, this yields
a valid new derivation tree. We call this rule
``\tagStack''.

By induction on the stack height, one can see that any
derivation tree with leaves in $N\cup T\cup\{\varepsilon\}$ can be
built in this way from atomic trees that consist of a single
production or a single node. 

\begin{example}\label{examplecomposition}
	\Cref{composingtrees} shows some steps for building the derivation tree
	from \cref{derivationtree} (for the grammar $G$ in
	\cref{examplegrammar}) using ``\tagPlug'' and ``\tagStack''. The tree
	$t_1$ is a single production, and $t_2$ is obtained from $t_1$ using
	``\tagStack'' with the stack symbol $g$. Then, $t_3$ is obtained from
	$t_2$ by plugging in a derivation tree for $W\grammarstep
	UB\grammarstep Ub$ (which itself is built using ``\tagPlug'') into the
	two $W$-leaves of $t_2$.  Furthermore, $t_4$ is then constructed from $t_3$
	using ``\tagStack'' with the stack symbol $f$. 

	The tree in \cref{derivationtree} can then be obtained from $t_4$ by first plugging in trees for $V\grammarsteps Ua$ (i.e.~``\tagPlug'') and then using ``\tagStack'' with stack symbol $\bot$, and plugging in $E\grammarstep \varepsilon$ .
\end{example}

\newcommand{\treeA}[4]{
\begin{scope}[shift={#2},#1]
	\Tree [ .\node (#3) {$R$};
        [ .{$U$} ]
	[ .{$U$} ] ]
\end{scope}
}

\newcommand{\treeB}[4]{
	\begin{scope}[shift={#2},#1]
	\Tree [ .\node(#3){$R$};
	[ .\node(#4){$Rg$};
	  [ .{$Ug$}
	    [ .{$W$} ] ]
	  [ .{$Ug$}
	    [ .{$W$} ] ] ] ]
	\end{scope}
}

\newcommand{\treeC}[4]{
\begin{scope}[shift={#2},#1]
\Tree 
[ .\node (#3) {$R$};
  [ .{$Rg$}
    [ .{$Ug$}
      [ .{$W$} 
	[ .{$U$} ] 
	[ .{$B$} 
	  [ .{$b$} ] ] ] ]
    [ .{$Ug$}
      [ .{$W$} 
	[ .{$U$} ] 
	[ .{$B$} 
	  [ .{$b$} ] ] ] ] ] ]
\end{scope}
}

\newcommand{\treeD}[4]{
\begin{scope}[shift={#2},#1]
\Tree 
[ .\node (#3) {$R$};
  [ .\node (#4) {$Rf$};
    [ .{$Rgf$}
      [ .{$Ugf$}
        [ .{$Wf$} 
          [ .{$Uf$} 
	    [ . {$V$} ] 
	  ]
          [ .{$Bf$} 
	    [ .{$b$} ]
	  ] 
	] 
      ] 
      [ .{$Ugf$}
        [ .{$Wf$} 
          [ .{$Uf$} 
	    [ .{$V$} ] ]
          [ .{$Bf$} 
	    [ .{$b$} ] ] ] ]
    ]
  ] 
]
\end{scope}
}

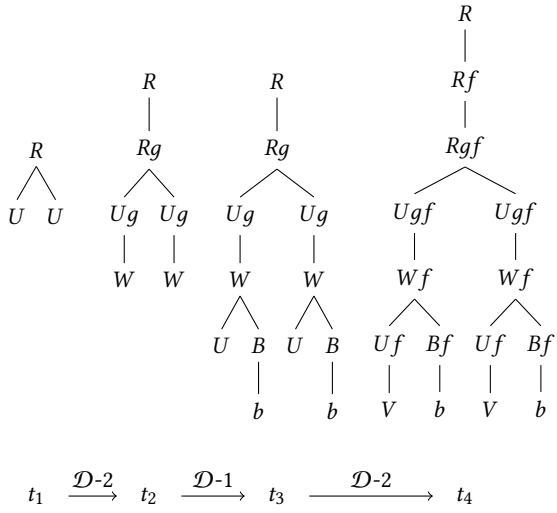
\begin{figure}
\begin{center}
\begin{tikzpicture}
\tikzset{level distance=25pt}

\coordinate (p) at (1.5cm,2cm);
\treeA{opacity=0}{(p)}{rootA}{refA}
\treeB{opacity=0}{(p)}{rootB}{refB}
\treeC{opacity=0}{(p)}{rootC}{refC}
\treeD{opacity=0}{(p)}{rootD}{refD}

\treeA{opacity=1}{(0,0)}{rootA2}{refA2}
	\treeB{opacity=1}{($($(rootA2.center)+(1.5cm,0cm)$)-($(refB.center)-(p)$)$)}{rootB2}{refB2}
	\treeC{opacity=1}{($($(rootB2.center)+(1.7cm,0cm)$)-($(rootB.center)-(p)$)$)}{rootC2}{refC2}
	\treeD{opacity=1}{($($(rootC2.center)+(2.5cm,0cm)$)-($(refD.center)-(p)$)$)}{rootD2}{refD2}

	\node (t1) at (rootA2.center |- ,-4.5) {$t_1$};
	\node (t2) at (rootB2.center |- ,-4.5) {$t_2$};
	\node (t3) at (rootC2.center |- ,-4.5) {$t_3$};
	\node (t4) at (rootD2.center |- ,-4.5) {$t_4$};

	\draw (t1) edge[->, shorten >=0.2cm, shorten <=0.2cm] node[above] {\tagStack} (t2);
	\draw (t2) edge[->, shorten >=0.2cm, shorten <=0.2cm] node[above] {\tagPlug} (t3);
	\draw (t3) edge[->, shorten >=0.2cm, shorten <=0.2cm] node[above] {\tagStack} (t4);
\end{tikzpicture}
\end{center}
	\caption{Example of building a derivation tree using transformations ``\tagPlug'' and ``\tagStack''. See \cref{examplecomposition} for details.}\label{composingtrees}
\end{figure}

	\subsubsection*{The operator $\Dop$}
	In order to describe the popping of all leaf non-terminals in parallel, we define the relation $\grammarstepi{f}$ for each index $f\in I$. For $\bu,\bv\in\N^{N\cup T}$, we write $\bu\grammarstepi{f}\bv$ if, informally, $\bv$ is obtained from $\bu$ by taking each non-terminal (occurrence) $A$ in $\bu$, picking a pop production $Af\to B$, and replacing this occurrence of $A$ by $B$. More formally, we have $\bu\grammarstepi{f} \bv$ if there are $\bw\in\N^T$ and productions $A_1f\to B_1,\ldots,A_nf\to B_n$ in $P$ such that
	\[ \bu=\bw+A_1+\cdots +A_n,\quad \bv=\bw+B_1+\cdots+B_n. \]
Recall that we identify each $a\in N\cup T$ with the vector $\bu_a\in\N^{N\cup T}$, where $\bu_a(a)=1$ and $\bu_a(b)=0$ for every $b\ne a$. Thus, $\bw+A_1+\cdots+A_n$ is the vector $\bw$, plus one occurrence of $A_i$ for each $i$ (and similar for $\bw+B_1+\cdots+B_n$).

With $\grammarstepi{f}$ in hand, we can define the operator $\Dop$:
\begin{align}
	\Dop(K)_A = \quad &K_A \notag \\
	       \cup ~&\{ \bu-B+\bv \mid \bu\in K_A,~\bu(B)\ge 1,~\bv\in K_B\} \label{derive-plug}\tag{\tagPlug}\\
	       \cup ~&\{ \bv\in\N^{N\cup T} \mid A\to Bf\in P,~\bu\in K_B,~\bu\grammarstepi{f}\bv \} \tag{\tagStack}\label{derive-stack}
\end{align}
Note that here, the two rules \eqref{derive-plug} and \eqref{derive-stack}
correspond exactly to the two transformations described above.

The final ingredient in our fixpoint characterization of $\Delta(G)$ is the starting value. The \emph{initial set tuple} is the set tuple $\iniTuple$ with 
\[	\iniTuple_A=\{A\}\cup \{w\in (N\cup T)^* \mid \text{there is a production $A\to w$ in $G$}\} \]
for each $A\in N$. We can now state our fixpoint characterization:
\begin{proposition}\label{fixpoint-indexed}
	$\Dop^\omega(\iniTuple)=\Delta(G)$.
\end{proposition}

We prove \cref{fixpoint-indexed} separately in two lemmas. First, we show that every vector built in the described way is derivable.
\begin{lemma}
$\Dop^\omega(\iniTuple)\subseteq\Delta(G)$.
\end{lemma}
\begin{proof}
	We have to argue that each vector produced using one of the rules
	\eqref{derive-plug} and \eqref{derive-stack} corresponds to a derived
	sentential form in $G$. Hence, we show that if
	$K\subseteq\Delta(G)$, then $\Dop(K)\subseteq\Delta(G)$.

	We will use the following notation. For a word $u\in (NI^*\cup T)^*$,
	let $u^f$ denote the word obtained from $u$ by adding $f$ at the bottom
	of each stack in $u$. In other words, if $u=u_0A_1x_1u_1\cdots
	A_nx_nu_n$, with $u_0,\ldots,u_n\in T^*$, $A_1,\ldots,A_n\in N$, and
	$x_1,\ldots,x_n\in I^*$, then $u^f:=u_0A_1x_1fu_1\cdots A_nx_nfu_n$.
	Observe that with this notation, $u\grammarsteps v$ implies
	$u^f\grammarsteps v^f$ for any $f\in I$.

	\begin{enumerate}
		\item If $\bu-B+\bv\in\N^{N\cup T}\in\Dop(K)_A$ is obtained
			using \eqref{derive-plug}, then we know that there are
			words $u,v\in (N\cup T)^*$ with $A\grammarsteps u$ and
			$B\grammarsteps v$ and $u$ contains an occurrence of
			$B$. Thus, $\bu-B+\bv$ is the Parikh image of a word
			derivable by replacing $B$ with $v$ in $u$.

		\item If $\bv\in\N^{N\cup T}$ is obtained using \eqref{derive-stack}, then there is a word $u$ with $B\grammarsteps u$ and $\bu=\Parikh{u}$ and $\bu\grammarstepi{f}\bv$. Then because of $\bu\grammarstepi{f}\bv$, there is a word $v\in (N\cup T)^*$ with $u^f\grammarsteps v$ and $\bv=\Parikh{v}$.
			This means, we have a derivation
			\[ A\grammarstep Bf \grammarsteps u^f \grammarsteps v, \]
			where $Bf\grammarsteps u^f$ holds because of our observation above.
			Thus, $\bv\in\Delta(G)_A$.\qedhere
	\end{enumerate}
\end{proof}

The second step is to show that every derivation can be decomposed into
transformations ``\tagPlug'' and ``\tagStack'':
\begin{lemma}
$\Dop^\omega(\iniTuple)\supseteq\Delta(G)$.
\end{lemma}
\begin{proof}
	For the converse, we proceed by induction on the maximal stack height
	at the nonterminals of a derivation.
	More precisely, we show by induction on $h$ that if $A\grammarsteps u$
	for some word $u\in (N\cup T)^*$ via a derivation that uses only stacks
	of height $\le h$, then $\Parikh{u}$ belongs to
	$\Dop^\omega(\iniTuple)_A$. If the derivation involves only zero-height
	stacks, then $u$ can be derived from $A$ using only context-free
	productions and it is easy to see that $\Parikh{u}$ can be built using
	\eqref{derive-plug}.

	Consider a derivation $A\grammarsteps u$ with maximal stack height
	$h>0$.  Roughly speaking, we show that a derivation tree $t$ for this
	derivation can be transformed into a tree $t'$ such that 
	\begin{enumerate}
		\item $t'$ and $t$ have the same yield $u\in (N\cup T)^*$, 
		\item the nodes of $t'$ have labels in $N\cup T\cup\{\emptyWord\}$, 
		\item each inner node is labeled by some $B\in N$ and its
			children form a word $w\in (N\cup T)^*$ such that
			$\Parikh{w}\in \Dop^\omega(\iniTuple)_B$.
	\end{enumerate}
	If this is established, then clearly one can use the rule
	\eqref{derive-plug} to place $\Parikh{u}$ in $\Dop^\omega(\iniTuple)_A$.
	The tree $t'$ will be obtained using contractions, which we describe first.
	A \emph{subtree of $t$} consists of a node $x$ and a set $U$ of nodes
	such that
	\begin{enumerate}
	\item every node in $U$ lies on a path from $x$ to a leaf, and
	\item every path from $x$ to a leaf visits $U$ exactly once.
	\end{enumerate}
	The subtree is \emph{complete} if all nodes in $U$ are leaves.  By
	\emph{contracting} a subtree consisting of $x$ and $U$, we mean taking
	all the complete subtrees whose root is in $U$ and making attaching
	them directly under $x$. In other words, we remove all the nodes
	between $x$ and $U$.

	Let us now describe how to obtain $t'$. Since in $t$, there are nodes with stack height $h>0$, $t$ must apply push productions to nodes with stack height zero. For each of those nodes $x$, we will describe a subtree rooted at $x$. Contracting all these subtrees then yields the desired $t'$.

	Consider a stack-height-zero node $x$ to which a push production is applied, say $B\to Cf$ with $B,C\in N$,
	Since $t$ derives
	a word in $(N\cup T)^*$, the $Cf$ node cannot be a leaf (it has a
	non-empty stack). 
	Since every leaf in $t$ belongs to $N\cup T\cup\{\varepsilon\}$, every
	path $\pi$ from $Cf$ to a leaf, contains a last node $\ell_\pi$ that still
	contains this occurrence of $f$. We define a set $U$ of nodes: For each node $\ell_\pi$:
	\begin{enumerate}
		\item If the production applied to $\ell_\pi$ is a pop production, then we include $\ell_\pi$ itself in $U$.
		\item If the production applied to $\ell_\pi$ is a terminal production, then we include all children of $\ell_\pi$ in $U$.
	\end{enumerate}
	Note that these are the only possible productions, because
	(i)~$\ell_\pi$ is the last node on $\pi$ to contain $f$ and (ii)~any
	context-free production that is not a terminal production has a right-hand side
	in $N^+$ and thus passes on the entire stack to all child nodes (recall that
	context-free productions have right-hand sides in $N^*\cup T^*$).

	Let $w$ be the word consisting of all
	node labels of nodes in $U$. Then $w$ is of the form 
	\[ w=w_0D_1fw_1\cdots D_nfw_n \]
	for some $w_0,\ldots,w_n\in (N\cup T)^*$, and $D_1,\ldots,D_n\in N$.
	Moreover, it clearly satisfies $Cf\grammarsteps w$. By the choice of
	$U$, the $f$ occurrence in $Cf$ is either never removed during this
	derivation or it is removed by a terminal production. Therefore, we
	also have a derivation
	\[ C\grammarsteps w_0D_1w_1\cdots D_nw_n=:w^{-f} \]
	with maximal stack height strictly smaller than $h$. By induction,
	this implies $\Parikh{w^{-f}}\in \Dop^{\omega}(\iniTuple)_C$.

	By the choice of $U$, we know that for each $D_if$ in $w$, the
	production applied to $D_if$ is a pop production.  Therefore, we denote
	by $w'$ the word obtained from $w$ by replacing each $D_if$ by the
	resulting nonterminal $E_i$, hence $w'=w_0E_1w_1\cdots E_nw_n$. Then we
	have $\Parikh{w^{-f}}\grammarstepi{f} \Parikh{w'}$ and thus, by
	\eqref{derive-stack}, we obtain $w'\in\Dop^\omega(\iniTuple)_B$.

	Let $U'$ be the set of nodes spelling $w'$. Then clearly, $x$ and $U'$
	form a subtree. Moreover, since $w'\in \Dop^\omega(\iniTuple)_B$, contracting
	$(x,U')$ will yield a node compatible with the requirements for $t'$. Hence, we
	perform this contraction for every stack-height-zero node $x$ to which a push
	production is applied, and arrive at the tree $t'$ as desired.
\end{proof}

\subsection{Slices as fixpoints}
It will be helpful to view slice closures as fixpoints, by way of the 
operator $\Sop\colon \powerset{\N^{k}}\to \powerset{\N^k}$ that intuitively performs one slice operation. For $K\subseteq\N^k$, we define
\[ \Sop(K)=\{\bu+\bv+\bw \mid \bu,\bv,\bw\in\N^{k},~\bu\in K,~\bu+\bv\in K,~\bu+\bw\in K \} \]
Note that by picking $\bv=\bw=\bzero$, we have $K\subseteq \Sop(K)$. Moreover, raising this operator to the power of $\omega$ as in \cref{operator-powers}, indeed yields that $\Sop^\omega(K)$ is the slice closure of $K\subseteq\N^k$.

To simplify notation, we apply much of the terminology about subsets of $\N^k$
to set tuples, usually component-wise. For set tuples
$K,L\in\powerset{\N^{N\cup T}}^N$, we write $K\subseteq L$ if $K_A\subseteq
L_A$ for each $A\in N$. A set tuple $K\in\powerset{\N^{N\cup T}}^N$ is called
\emph{slice} if $K_A\subseteq\N^{N\cup T}$ is a slice for each $A\in N$.
Moreover, if $K$ is a set tuple, we set $\Sop(K)_A=\Sop(K_A)$ (i.e.\ we apply
$\Sop$ component-wise).  Accordingly, for a set tuple $K\in\powerset{\N^{N\cup
T}}^N$, we call $\Sop^\omega(K)$ the \emph{slice closure} of $K$.
A set tuple $K$ is \emph{semilinear} if each set $K_A\subseteq\N^{N\cup T}$ is semilinear. It is \emph{effectively semilinear} if we can compute a semilinear representation (it will always be clear from the context how $K$ is given).
\subsection{Computing slice closures}

In our algorithm to compute the slice closure of indexed languages, we will proceed in two steps. We will first compute the slice closure of the tuple $\Delta(G)$, which contains all Parikh images of sentential forms of $G$. In the second step, we then compute the slice closure of $\Parikh{L}$ from the slice closure of $\Delta(G)$. Here, it is key that the slice closure of $\Parikh{L}=\Delta(G)_S\cap\N^T$ can be reconstructed from the slice closure of $\Delta(G)_S$. More specifically, we can compute it by way of $\Sop^{\omega}(\Parikh{L})=\Sop^\omega(\Delta(G)_S)\cap\N^T$, because of Lemma \ref{intersection-commutes}. We use the convention that for any set $U\subseteq \N^{N\cup T}$, the intersection $U\cap \N^T$ refers to the tuples where all the coordinates indexed by $N$ are zero.%
\begin{lemma}\label{intersection-commutes}
	For every set $U\subseteq\N^{N\cup T}$, we have $\Sop^\omega(U\cap \N^T)=\Sop^\omega(U)\cap \N^T$.
\end{lemma}
\begin{proof}
	The inclusion ``$\subseteq$'' holds because $\Sop^\omega(U\cap \N^T)$ is
	clearly included in $\Sop^\omega(U)$ and in $\N^T$.  

	For ``$\supseteq$'', we first claim that $\Sop(U)\cap
	\N^T\subseteq\Sop(U\cap \N^T)$.  Indeed, for each element $\bm{u}+\bm{v}+\bm{w}$
	added by $\Sop$ with $\bm{u}+\bm{v}+\bm{w}\in \N^T$, the vectors $\bm{u}$,
	$\bm{u}+\bm{v}$, and $\bm{u}+\bm{w}$ must clearly already belong to $\N^T$, hence $\bm{u}+\bm{v}+\bm{w}\in\Sop(U\cap \N^T)$.
	This implies $\Sop^i(U)\cap \N^T\subseteq\Sop^i(U\cap \N^T)$ by induction on
	$i$, and thus $\Sop^{\omega}(U)\cap \N^T\subseteq \Sop^\omega(U\cap \N^T)$.
\end{proof}

\mysubsection{Slice closure vs. other closure operators}
This is a good moment to highlight why the notion of the slice closure is so important in our proofs---even if one just wants to show \cref{affine-hull}: One might wonder whether \cref{affine-hull} could be obtained by somehow computing the affine hull of $\Delta(G)_S=\Dop^\omega(\iniTuple)_S$ and then use that to compute the affine hull of $\Parikh{L}$. However, we will now illustrate that the affine hull of $\Delta(G)_S$ is not enough, even for just computing the affine hull of $\Parikh{L}=\Delta(G)_S\cap\N^T$. 

We even show something slightly stronger: Even if we could compute the Zariski closure (which encodes more information than the affine hull) of $\Delta(G)_S$, this would not contain enough information to reconstruct the affine hull of $\Parikh{L}$. Let us first define the Zariski closure. A subset $V\subseteq\Q^k$ is an \emph{affine variety} if it is of the form
\[ \{\bv\in\Q^k \mid p_1(\bv)=\cdots=p_n(\bv)=0 \},\]
where $p_1,\ldots,p_n\in\Q[X_1,\ldots,X_k]$ are finitely many polynomials. It
is a simple consequence of Hilbert's Basis Theorem that the intersection of any
collection of affine varieties is again an affine variety. Therefore, we can
assign to every subset $U\subseteq\Q^k$ the smallest affine variety including
$U$, denoted $\overline{U}$. The closure $\overline{U}$ is also called the
\emph{Zariski closure} of $U$.

Suppose we have $N=\{A\}$ and $T=\{a,b\}$. Then we identify $\N^{N\cup T}$ with $\N^3$, where the left-most entry corresponds to $A$, the second to $a$, and the third to $b$. Consider the set 
\[ U=\{(x,y,z)\in\N^3 \mid x\ge 1\}. \]
We claim that then for every subset $W\subseteq\{0\}\times\N^2$, we have $\overline{U\cup W}=\overline{U}$. This is because every polynomial $p\in\Q[X_1,X_2,X_3]$ that vanishes on $U$ must also vanish on $W$: for any $y,z\in\N$, we know that $p(X_1,y,z)\in\Q[X_1]$ must have every positive natural number as a root and thus be identically zero; meaning $p$ also vanishes on $W$. However, this means $\overline{U\cup W}$ does not encode enough information to recover anything about $W$. \cref{intersection-commutes}, on the other hand, tells us that $\Sop^\omega(U\cup W)\cap (\{0\}\times\N^2)=\Sop^\omega(W)$, meaning $\Sop^\omega(W)$ can be reconstructed.

\begin{lemma}\label{derivation-step}
	Given a semilinear set $K\in\powerset{\N^{N\cup T}}^N$, the set tuple $\Dop(K)$ is effectively semilinear.
\end{lemma}
\begin{proof}
	Suppose we are given a Presburger formula $\varphi_A$ for each set
	$K_A\subseteq\N^{N\cup T}$.  We provide a Presburger formula for each
	of the rules \eqref{derive-plug} and \eqref{derive-stack}. For \eqref{derive-plug},
	this is easy, because the set of all $\bu-B+\bv$ with $\varphi_A(\bu)$
	and $\varphi_B(\bv)$ and $\bu(B)\ge 1$ is clearly Presburger-definable.

	For the rule \eqref{derive-stack}, it suffices to show that the
	relation $\grammarstepi{f}$ is Presburger-definable for each $f\in I$.
	To this end, we introduce for each production $Cf\to D$ with $C,D\in N$ 
	a variable that counts to how many occurrences of $C$ in $\bu$ the production $Cf\to D$ is applied. 
	Let $C_if\to D_i$ for $i=1,\ldots,n$ be all the productions that pop $f$
	and introduce the variables $x_1,\ldots,x_n$.
	Moreover, we introduce a variable $\by$ ranging over $\N^T$ that holds all the terminal occurrences in $\bu$, which are unchanged between $\bu$ and $\bv$.
	Observe that then $\bu\grammarstepi{f} \bv$ if and only if
	\begin{multline*}
	\exists x_1,\ldots,x_n\in\N,~\by\in\N^T\colon\bigwedge_{E\in N} \bu(E)=\sum_{i\in\{1,\ldots,n\},~C_i=E} x_i ~\wedge~  \\
	\bigwedge_{E\in N} \bv(E)=\sum_{i\in\{1,\ldots,n\},~D_i=E} x_i ~\wedge~
		\bigwedge_{a\in T} \bv(a)=\bu(a), 
	\end{multline*}
	which is clearly expressible in Presburger arithmetic.
\end{proof}

The following \lcnamecref{slice-closure-semilinear} says that the slice closure
of semilinear sets is effectively semilinear. This follows using ideas of
Grabowski~\cite[corollary to Theorem 1]{DBLP:conf/fct/Grabowski81}. To make the
paper more self-contained, we include a short proof.
\begin{theorem}\label{slice-closure-semilinear}
	Given a semilinear $K\subseteq\N^k$, the set $\Sop^{\omega}(K)$ is effectively semilinear.
\end{theorem}
\begin{proof}
	Roughly speaking, we enumerate semilinear representations and check whether they satisfy a particular property. Once a semilinear representation satisfies the property, we know that it represents $\Sop^\omega(K)$.
	More precisely, the algorithm enumerates pairs $(R,m)$, where $R=(\bb_1,P_1;\ldots;\bb_n,P_n)$ is a semilinear representation and $m\in\N$. We call the pair $(R,m)$ \emph{good for $K$} if
	\begin{enumerate}
		\item\label{semilinear-slice-large} $K\subseteq \langle R\rangle$,
		\item\label{semilinear-slice-slice} $\langle R\rangle$ is a slice, and
		\item\label{semilinear-slice-small} each vector $\bb_i$ and each finite set $\bb_i+P_i$ is included in $\Sop^m(K)$.
	\end{enumerate}
	If our algorithm encounters a pair $(R,m)$ that is good for $K$, it
	returns $R$. Note that it is decidable whether a given pair is good for
	$K$, because given a Presburger formula for $K$ and a pair $(R,m)$, the
	fact that $(R,m)$ is good for $K$ can be expressed in Presburger
	arithmetic. Here, we use the fact that for a given $m$, we can
	construct a Presburger formula for the set $\Sop^m(K)$.

	The algorithm is sound, i.e.\ if $(R,m)$ is good for $K$, then $\langle
	R\rangle=\Sop^\omega(K)$:
	\cref{semilinear-slice-large,semilinear-slice-slice} imply that
	$\langle R\rangle$ is a slice that includes $K$. Moreover,
	\cref{semilinear-slice-small} implies $\langle R\rangle\subseteq\Sop^\omega(K)$,
	because the entire set $\bb_i+\langle P_i\rangle$ can be obtained from
	$\bb_i$ and $\bb_i+P_i$ by repeatedly applying the slice operation.

	It remains to show that the algorithm terminates. This follows from the
	fact that $\Sop^\omega(K)$ is semilinear: Take a semilinear
	representation $R=(\bb_1,P_1;\ldots;\bb_n,P_n)$ for $\Sop^\omega(K)$.
	Then clearly $K\subseteq\langle R\rangle$ and $\langle R\rangle$ is a
	slice.  Moreover, since $\{\bb_i\}\cup(\bb_i+P_i)$ is a finite subset of
	$\Sop^\omega(K)=\bigcup_{i\ge 0}\Sop^i(K)$, there must be an $m$ such
	that $\Sop^m(K)$ contains this finite set. Hence, the pair $(R,m)$ is
	good for $K$.
\end{proof}

\begin{proposition}\label{reorder}
	For every set tuple $K$, we have $\Sop^{\omega}(\Dop^{\omega}(K))=(\Sop^{\omega}\circ\Dop)^\omega(K)$.
\end{proposition}

Before we prove \cref{reorder}, let us see how it is used in the proof of \cref{slice-closures}.
\begin{proof}[Proof of \cref{slice-closures}]
	Our algorithm computes a sequence 
	\begin{equation} K^{(0)}\subseteq K^{(1)}\subseteq K^{(2)}\subseteq \cdots \label{set-tuple-sequence}\end{equation}
	of semilinear set tuples with $K^{(0)}=\iniTuple$ and
	$K^{(i+1)}=\Sop^{\omega}(\Dop(K^{(i)}))$ for $i\ge 0$. Note that due to
	\cref{derivation-step,slice-closure-semilinear}, we can compute a
	Presburger formula for $K^{(i+1)}$ from a 
	Presburger formula for $K^{(i)}$. The algorithm continues while $K^{(i)}\subsetneq K^{(i+1)}$ and terminates once we have $K^{(i+1)}=K^{(i)}$.

	Notice that upon termination at $K^{(i)}$, we have
	$K^{(i)}=(\Sop^\omega\circ\Dop)^\omega(\iniTuple)$, which equals
	$\Sop^\omega(\Dop^\omega(\iniTuple))$ by \cref{reorder}. Therefore, by
	\cref{intersection-commutes}, the slice closure of the indexed language is
	$K^{(i)}_S\cap \N^T$, for which we can clearly compute a Presburger formula.
	Thus, it remains to show that the algorithm terminates.

	Since ascending chains of slices become stationary, we know that for
each non-terminal $A$, the sequence $K_A^{(0)}\subseteq K_A^{(1)}\subseteq
K_A^{(2)}\subseteq\cdots$ becomes stationary at some index $m_A$. Beyond all
these finitely many indices, the sequence $K^{(0)}\subseteq K^{(1)}\subseteq
K^{(2)}\subseteq\cdots$ of set tuples must be stationary in all components.
\end{proof}

The remainder of this section is now devoted to proving \cref{reorder}.
The inclusion ``$\subseteq$'' is obvious, because we clearly have
$\Sop^\omega(\Dop^\omega(K))=\bigcup_{i\ge 0} \Sop^i(\Dop^i(K))$ and moreover $\Sop^i(\Dop^i(K))\subseteq (\Sop\circ\Dop)^{2i}(K)\subseteq (\Sop^\omega\circ\Dop)^{2i}(K)$. The converse inclusion ``$\supseteq$'' will easily follow from:
\begin{lemma}\label{commute}
	For every set tuple $K$, we have $\Dop(\Sop(K))\subseteq \Sop^2(\Dop(K))$.
\end{lemma}
\begin{proof}
	For any element of $\Dop(\Sop(K))_A$, we distinguish between the rule
	applied for inclusion in $\Dop(\Sop(K))$, i.e.\ \eqref{derive-plug} or
	\eqref{derive-stack}. First, suppose \eqref{derive-plug} was applied.
	Then our element of $\Dop(\Sop(K))_A$ is of the form $\bu-B+\bv$ with
	$\bu\in\Sop(K)_A$, $B\in N$, $\bu(B)\ge 1$, and $\bv\in\Sop(K)_B$. This
	means, we can write
	\begin{align*}
		\bu&=\bs+\bx+\by & &\text{for $\bs,\bs+\bx,\bs+\by\in K_A$}, \\
		\bv&=\bs'+\bx'+\by' & &\text{for $\bs',\bs'+\bx',\bs'+\by'\in K_B$},
	\end{align*}
	for some $\bs,\bx,\by,\bs',\bx',\by'\in\N^{N\cup T}$. Note that this is
	even true if $\bu\in K_A$, $\bv\in K_B$, because we allow
	$\bx,\by,\bx',\by'$ to be zero. 

	Since $\bu(B)\ge 1$, we know that $\bs(B)\ge 1$ or $\bx(B)\ge 1$, or
	$\by(B)\ge 1$. Here, the case $\by(B)\ge 1$ is symmetric to $\bx(B)\ge 1$, so suppose $\bs(B)\ge 1$ or $\bx(B)\ge 1$.
	In either case, we have $(\bs+\bx)(B)\ge 1$ and therefore
	\begin{align}
		&\bs+\bx-B+\bs'      & &\in \Dop(K)_A \label{commute-a}\\
		&\bs+\bx-B+\bs'+\bx' & &\in \Dop(K)_A \label{commute-b}\\
		&\bs+\bx-B+\bs'+\by' & &\in \Dop(K)_A \label{commute-c}
	\end{align}
	by \eqref{derive-plug}. Now \cref{commute-a,commute-b,commute-c} together yield
	\begin{align} \bs+\underbrace{\bx-B+\bs'+\bx'+\by'}_{=:\bz} \in \Sop(\Dop(K))_A. \label{commute-d}\end{align}
	Now we distinguish between the cases $\bs(B)\ge 1$ and $\bx(B)\ge 1$.
	\begin{enumerate}
		\item Suppose $\bx(B)\ge 1$. In this case, the vector $\bz$ in
			\cref{commute-d} belongs to $\N^{N\cup T}$ (because
			$\bx-B$ does). Therefore, combined with $\bs,\bs+\by\in
			K_A\subseteq\Sop(\Dop(K))_A$, we obtain
			\[ \bu-B+\bv=\bs+\bx+\by-B+\bs'+\bx'+\by' \in \Sop^2(\Dop(K))_A, \]
		as desired.
	\item Suppose $\bs(B)\ge 1$. Then we can apply the reasoning for \cref{commute-d}, but (a)~with $\bs$ instead of $\bs+\bx$ and again (b)~with $\bs+\by$ instead of $\bs+\bx$ to obtain
	\begin{align*} 
		&\bs-B+\bs'+\bx'+\by' &&\in \Sop(\Dop(K))_A,\\
		&\bs+\by-B+\bs'+\bx'+\by' &&\in \Sop(\Dop(K))_A,
	\end{align*}
	which together with \cref{commute-d} yields
	\[ \bu-B+\bv=\bs+\bx+\by-B+\bs'+\bx'+\by' \in \Sop^2(\Dop(K))_A, \]
	as desired.
	\end{enumerate}

	It remains to consider the rule \eqref{derive-stack}. Consider a
	production $A\to Bf$ and vectors $\bu\in \Sop(K)_B$ and $\bv\in\N^{N\cup T}$
	with $\bu\grammarstepi{f} \bv$. Then we can write $\bu=\bs+\bx+\by$ such that
	$\bs,\bs+\bx,\bs+\by\in K_B$. Observe that we can then decompose $\bv=\bs'+\bx'+\by'$ such that
	$\bs\grammarstepi{f}\bs'$, $\bx\grammarstepi{f}\bx'$, and $\by\grammarstepi{f}\by'$. This %
	implies:
	\begin{align*}
		&\bs'     &&\in \Dop(K)_A &&\text{since $\bs\grammarstepi{f}\bs'$} \\
		&\bs'+\bx'&&\in \Dop(K)_A &&\text{since $\bs+\bx\grammarstepi{f}\bs'+\bx'$} \\
		&\bs'+\by'&&\in \Dop(K)_A &&\text{since $\bs+\by\grammarstepi{f}\bs'+\by'$}
	\end{align*}
	and therefore $\bv=\bs'+\bx'+\by'\in\Sop(\Dop(K))_A\subseteq\Sop^2(\Dop(K))_A$.
\end{proof}
Now the inclusion ``$\supseteq$'' in \cref{reorder} follows easily: With $i$
applications of \cref{commute}, we see that
$\Dop(\Sop^i(K))\subseteq\Sop^{2i}(\Dop(K))$ for every $i\ge 0$ and thus
$\Dop(\Sop^\omega(K))\subseteq\Sop^{\omega}(\Dop(K)))$.  By induction on $j$, this
implies $(\Sop^\omega\circ\Dop)^j(K)\subseteq \Sop^\omega(\Dop^j(K))$ for every
$j\ge 0$.  Thus, we obtain
$(\Sop^\omega\circ\Dop)^\omega(K)\subseteq\Sop^\omega(\Dop^\omega(K))$.

\section{Word equations with counting constraints}\label{sec:wordeqs}
Let $\Omega$ be a finite collection of
variables and $A$ a finite set of letters, and consider  the word equation $(U,V)$ over $(A,\Omega)$.
Recall that a \emph{solution} to $(U,V)$  is a morphism
$\sigma\colon(\Omega \cup A)^*\to A^*$ that fixes $A$ point-wise and such that $\sigma(U)=\sigma(V)$.

A \emph{word equation with rational constraints} is a word equation $(U,V)$
together with a regular language $R_X\subseteq A^*$ for each variable
$X\in\Omega$. In this situation, we say that $\sigma\colon\Omega\to A^*$ is a
\emph{solution} if it is a solution to the word equation $(U,V)$ and also
satisfies $\sigma(X)\in R_X$ for each $X\in\Omega$.

\newcommand{\enc}[1]{\mathsf{enc}(#1)}
\newcommand{\cT}{\mathcal{T}}
\mysubsection{Representing solutions} We will rely on a result that provides a language-theoretic description of the set of solutions of an equation, which requires a representation as words. Given a map $\sigma\colon\Omega\to A^*$ with $\Omega=\{X_1,\ldots,X_k\}$, we define
\[ \enc{\sigma}~ = ~\sigma(X_1)\#\cdots\#\sigma(X_k), \]
where $\#$ is a fresh letter not in $A$. Hence, $\enc{\sigma}\in(A\cup\{\#\})^*$.

\mysubsection{Counting functions}
Counting constraints will be expressed using counting functions, which are functions
$f\colon (A\cup\{\#\})^*\to \Z^n$ that are rational when viewed as a relation $\subseteq (A\cup\{\#\})^*\times\Z^n$.
A \emph{counting transducer} is a finite automaton $\cT=(Q,\Sigma,E,q_0,F)$, where $Q$ is a finite set of \emph{states}, $\Sigma$ is its \emph{input alphabet}, $E\subseteq Q\times\Sigma^*\times\Z^n\times Q$ is its \emph{edge relation}, $q_0\in Q$ is its \emph{initial state}, and $F\subseteq Q$ is its set of \emph{final states}. A \emph{run} is a sequence
\[ q_0(w_1,\bx_1)q_1\cdots (w_m,\bx_m)q_m \]
with $q_0,\ldots,q_m\in Q$, $w_1,\ldots,w_m\in\Sigma^*$, and $\bx_1,\ldots,\bx_m\in\Z^n$ such that each $(q_i,w_{i+1},\bx_{i+1},q_{i+1})\in E$, and $q_m\in F$. The relation \emph{defined} by the counting transducer, denoted $R(\cT)$ is the set of all pairs $(w_1\cdots w_m,\bx_1+\cdots+\bx_m)\in \Sigma^*\times\Z^n$ for runs as above. The counting transducer is called \emph{functional} if for every word $w\in\Sigma^*$, there is at most one vector $\bx\in \Z^n$ such that $(w,\bx)\in R(\cT)$. A function $f\colon \Sigma^*\to\Z^n$ is called a \emph{counting function} if its graph $\{(w,f(w)) \mid w\in\Sigma^*\}\subseteq\Sigma^*\times\Z^n$ is defined by a functional counting transducer. In particular, the graph of $f$ is a rational subset of  $\Sigma^*\times\Z^n$

We will apply counting functions to encodings of solutions. Let us see some examples.
\begin{enumerate}
\item Letter counting: The function $f_{X,a}$ with $f_{X,a}(\enc{\sigma})=|\sigma(X)|_a$ for some letter $a\in A$ and variable $X\in\Omega$. Here,  the transducer increments the counter for each $a$ between the $i$-th and $(i+1)$-st occurrence of $\#$, where $X=X_i$.
\item Counting positions with MSO properties. Consider monadic second-order logic (MSO) formulas evaluated in the usual way over finite words. Suppose we have an MSO formula $\varphi(x)$ with one free first-order variable $x$. Then we can define the function $f_{\varphi}$ such that $f_{\varphi}(\enc{\sigma})$ is the number of positions $x$ in $\enc{\sigma}$ where $\varphi(x)$ is satisfied. Then $f_{\varphi}$ is a counting function, which follows from the fact that MSO formulas define regular languages. For example, we could count the number of $a$'s such that there is no $c$ between the $a$ and the closest even-distance $b$.

\item Linear combinations: If $f_1,\ldots,f_m\colon\Sigma^*\to\Z^n$ are counting functions, then so is $f$ with $f(w)=\lambda_1f_1(w)+\cdots+\lambda_mf_m(w)$ for some $\lambda_1,\ldots,\lambda_m\in\Z$. This can be shown with a simple product construction.
\item\label{length-function} Length function: The function $L_X$ with $L_X(\enc{\sigma})=|\sigma(X)|$ for some $X\in\Omega$. For this, we can take a linear combination of letter counting functions.
\end{enumerate}

The following is the problem of \emph{word equations with rational and counting constraints}:
\begin{description}
\item[Given] A word equation $(U,V)$ with rational constraints and a counting function $f\colon (A\cup\{\#\})^*\to\Z^n$
\item[Question] Is there a solution $\sigma$ that satisfies $f(\enc{\sigma})=0$?
\end{description}
As B\"{u}chi and Senger showed~\cite[Corollary 4]{buchi1990definability},
this problem in full generality is undecidable. However, it remains a well-known open problem
whether this problems is decidable when we restrict to \emph{length
constraints}. We say that $f\colon (A\cup\{\#\})^*\to\Z^n$ is a \emph{length
constraint} if each projection to a component of $\Z^n$ is a linear combination
of length functions as in \cref{length-function} above. Thus, a length
constraint can express that $\sigma(X)$ has the same length as $\sigma(Y)$ for
$X,Y\in\Omega$. The problem of \emph{word equations with rational and length
constraints} is the following:
\begin{description}
\item[Given] A word equation $(U,V)$ with rational constraints and a length constraint $f\colon (A\cup\{\#\})^*\to\Z^n$
\item[Question] Is there a solution $\sigma$ that satisfies $f(\enc{\sigma})=\bzero$?
\end{description}

We will show here that a restriction of word equations with rational
constraints and counting constraints is decidable. In this restriction, we
require $f(\enc{\sigma})\ne\bzero$ rather than $f(\enc{\sigma})=\bzero$. The
problem of \emph{word equations with rational constraints and counting
inequations} is the following:
\begin{description}
\item[Given] A word equation $(U,V)$ with rational constraints and a counting function $f\colon (A\cup\{\#\})^*\to\Z^n$
\item[Question] Is there a solution $\sigma$ that satisfies $f(\enc{\sigma})\ne \bzero$?
\end{description}
Here, we show that word equations with rational constraints and counting inequations are decidable (\cref{counting-inequations-decidable}).

Before we come to the proof, let us mention that \cref{counting-inequations-decidable} yields a decidable fragment of word equations with length constraints. This is because in the special case that $f$ is a length constraint, the inequation $f(\enc{\sigma})\ne\bzero$ can be expressed using length \emph{equations} by standard tricks: Suppose $f=\lambda_1f_1+\cdots+\lambda_mf_m$ is a linear combination of length functions $f_1,\ldots,f_m$, with $\lambda_1,\ldots,\lambda_k>0$ and $\lambda_{k+1},\ldots,\lambda_m<0$, where $f_i(\enc{\sigma})=|\sigma(X_i)|$ for variable $X_i\in\Omega$. Then $\sigma$ is a solution with $f(\enc{\sigma})\ne\bzero$ if and only if there is a letter $a\in A$ and a system with
\begin{align*} X=X_1^{\lambda_1}\cdots X_k^{\lambda_k},&&Y=X_{k+1}^{|\lambda_{k+1}|}\cdots X_m^{|\lambda_m|}, && {|X|=|YaZ}| \end{align*}
or
\begin{align*} X=X_1^{\lambda_1}\cdots X_k^{\lambda_k},&&Y=X_{k+1}^{|\lambda_{k+1}|}\cdots X_m^{|\lambda_m|},&& |Y|=|XaZ|, \end{align*}
depending on whether $f(\enc{\sigma}) >\bzero$ or $f(\enc{\sigma}) < \bzero$. where $X,Y,Z$ are fresh variables. This can easily be expressed by word
equations with length constraints (relying on the standard fact that several
equations for words can be translated into a single equation). Thus,
\cref{counting-inequations-decidable} yields in particular a decidable
non-trivial fragment of word equations with length constraints.
Similarly, counting inequations can always be encoded in word equations with counting constraints.

One can view \cref{counting-inequations-decidable} as showing that it is
decidable whether \emph{all solutions} $\sigma$ to a word equation $(U,V)$ with
rational constraints satisfy $f(\enc{\sigma})=\bzero$ for some given counting
function $f$. One might wonder whether in all these cases, satisfaction of
$f(\enc{\sigma})=\bzero$ for all solutions $\sigma$ is already implied by the
equality $\Parikh{\sigma(U)}=\Parikh{\sigma(V)}$ and the rational constraints.
The following example shows that this is not the case.

\begin{example}
	Suppose $\Omega=\{X_1,X_2,X_3\}$, $A=\{a,b,c\}$, and consider the word equation
	\[ abX_1cX_1ba=X_1abcX_2X_3 \]
	with rational constraints $X_1,X_3\in\{a,b\}^*$, and
	$X_2\in\{ab\}^*ba$. Under these constraints, the equation is equivalent
	to the pair of equations $abX_1=X_1ab$ and $X_1ba=X_2X_3$. It is a
	basic fact from word combinatorics (see e.g.~\cite[Exercise
	2.5]{berstel79}) that the equation $abX_1=X_1ab$ is equivalent to
	$X_1\in\{ab\}^*$. This means, $X_1ba=X_2X_3$ and $X_2\in\{ab\}^*ba$
	imply $|X_3|=0$. Thus, the length constraint
	$f\colon(A\cup\{\#\})^*\to\Z$ with $f(u\#v\#w)=|w|$ satisfies
	$f(\enc{\sigma})=0$ for every solution $\sigma$. However, if we
	take the equation
	\[ abX_1cX_1ba=abX_1cX_2X_3 \]
	instead (with the same rational constraints), then $abX_1=abX_1$ is
	satisfied for any $X_1$. In particular, for every $k\ge 1$, we have a
	solution $\sigma$ with $\sigma(X_1)=ab(ba)^k$, $\sigma(X_2)=abba$,
	$\sigma(X_3)=(ba)^{k}$, where $f(\enc{\sigma})=2k$. This shows that the
	satisfaction the satisfaction of counting constraints for all solutions
	of (U,V) doesn't depend only on $(\Parikh{U},\Parikh{V})$ and the
	rational constraints imposed on the variables. 
\end{example}

The remainder of this section is devoted to proving
\cref{counting-inequations-decidable}. 

\mysubsection{Transductions}
First, we need a standard notion from language theory. A \emph{finite-state
transducer} is defined like a counting transducer, except that instead of the
$\Z^n$ component on the edges, we have words over some output alphabet
$\Gamma^*$. It recognizes a relation $R\subseteq\Sigma^*\times\Gamma^*$ which
is defined in the same way as for counting transducers (except that instead of
adding output vectors, we concatenate output words). A relation
$R\subseteq\Sigma^*\times\Gamma^*$ is \emph{rational} if it is recognized by
some finite-state transducer. For a language $L\subseteq\Sigma^*$, we define
$RL:=\{v\in\Gamma^* \mid \exists (u,v)\in R\colon u\in L\}$. It is well-known
that indexed languages are (effectively) closed under applying rational
transductions. In other words, for every indexed language $L$, we can compute
an indexed grammar for the language $RL$.

\mysubsection{Solution sets of word equations}
Our proof crucially relies on a structural description of solution sets of word
equations with rational constraints from \cite[Theorem
1]{DBLP:conf/icalp/CiobanuDE15}:
\begin{theorem}[Ciobanu, Elder, Diekert 2015]\label{solutions-edt0l}
For every word equation $(U,V)$ with rational constraints, the language
\begin{equation} \{\enc{\sigma} \mid \text{$\sigma\colon \Omega\to A^*$ is a solution} \} \label{solution-set}\end{equation}
is an EDT0L language.
\end{theorem}
Here, we do not need the exact definition of an EDT0L language, it suffices to
know that the EDT0L languages are (effectively) a subclass of the indexed
languages. This allows us to prove \cref{counting-inequations-decidable}.
\begin{proof}[Proof of \cref{counting-inequations-decidable}]
Given a word equation $(U,V)$ with rational constraints and a counting function $f\colon (A\cup\{\#\})^*\to\Z^n$, we first construct an indexed grammar for the language \cref{solution-set}, which we denote $L$. Next, from $f$, we construct a finite-state transducer for a relation $R_f\subseteq (A\cup\{\#\})^*\times\Gamma^*$ with $\Gamma=\{a_1,\bar{a}_1,\ldots,a_n,\bar{a}_n\}$ such that for every $u\in (A\cup\{\#\})^*$, there is a word $v\in\Gamma^*$ such that
\[ f(u)=(|v|_{a_1}-|v|_{\bar{a}_1},\ldots,|v|_{a_n}-|v|_{\bar{a}_n}) \]
and $(u,v)\in R_f$. This is easy: Given the counting transducer for $f$, we just
turn each edge labeled $(w,(x_1,\ldots,x_n))$ into one labeled
$(w,a_1^{(x_1)}\cdots a_n^{(x_n)})$, where $a_i^{(x_i)}$ is defined as $a_i^{x_i}$ if
$x_i\ge 0$ and as $\bar{a}_i^{|x_i|}$ if $x_i<0$. Now clearly, our instance of word equations with rational constraints and counting inequations is \emph{negative} if and only if $\Parikh{R_fL}\subseteq S$, where  
	\[ S=\{\bz\in\N^{\Gamma} \mid \bz(a_1)=\bz(\bar{a}_1),~\ldots,~\bz(a_n)=\bz(\bar{a}_n) \}. \]
Notice that $S$ is a slice, since it is the solution set of a homogeneous
linear equation. Therefore, our instance is negative if and only if
$\Sop^\omega(\Parikh{R_fL})\subseteq S$, which we can decide because we can
construct an indexed grammar for $R_fL$ and \cref{slice-closures} allows us to
compute a Presburger formula for $\Sop^\omega(\Parikh{R_fL})$.
\end{proof}

\subsubsection*{Remark on the proof of \cref{counting-inequations-decidable}}
It should be noted that if $L\subseteq\{a_1,\bar{a}_1,...,a_n,\bar{a}_n\}$ is
an EDT0L language, one can use affine programs (i.e.\ the results by
Karr~\cite{DBLP:journals/acta/Karr76} and Müller-Olm and
Seidl~\cite{DBLP:conf/icalp/Muller-OlmS04}) or weighted automata over a field~\cite[p.143--145]{Eilenberg1974} (see also~\cite[Chapter 8]{BojanczykCzerwinski2018}) to decide whether for every word in $L$, the number of
$a_i$ is the same as the number of $\bar{a}_i$. (As far as we can tell, this
has not been observed in the literature, but it is not difficult to show.)
However, it is crucial in our proof that we can do this for the language
$R_fL$, which is obtained using the rational transduction $R_f$ from an EDT0L language,
and EDT0L languages are not closed under rational transductions (see, e.g.\
\cite[Lemma 3.7]{ehrenfeucht1976relationship}), whereas indexed languages are. It seems possible that for
some (but most likely not all) counting functions $f$, one can modify the proof of
\cite[Thm.~1]{DBLP:conf/icalp/CiobanuDE15} to directly yield an EDT0L language
for $R_fL$, and then apply the aforementioned methods. This remains to be
explored.

\section{Conclusion}
This work initiates the study of slices closures of indexed languages, and thus
models of higher-order recursion. There are many questions that arise. The most
immediate question is whether slice closure computation (or decidability of
\eqref{kobayashi-inclusion}) is also possible for safe or unsafe HORS of orders $\ge 3$.

Second: What is the complexity of computing
slice closures of indexed languages? The algorithm presented here is a
saturation algorithm where each step involves an enumeration-based procedure for
computing the slice closure of a semilinear set. The first challenge is the lack
of complexity bounds for computing the slice closure of a semilinear
set. However, even with such a bound, it would not be clear how to bound the saturation. Since the ascending chain condition
for slices (and thus termination of our algorithm) ultimately results from
Dickson's lemma, it seems tempting to use length function theorems for vector
orderings (e.g.~\cite[Thm.~3.15]{Schmitz17}) to at least obtain an Ackermannian
upper bound. In~\cite[Prop.~4.8]{DBLP:conf/icalp/GanardiMPSZ22}, such an
approach has been developed for ascending chains of \emph{congruences}, but it
is not clear how to extend this to slices. 

Finally: Can we decide whether every vector in $\Parikh{L}$
satisfies a given polynomial equation? For affine programs, this is
decidable~\cite{DBLP:journals/ipl/Muller-OlmS04}, and one can even compute the
Zariski closure~\cite{DBLP:journals/jacm/HrushovskiOPW23}.

\begin{acks}
\raisebox{-9pt}[0pt][0pt]{\includegraphics[height=.8cm]{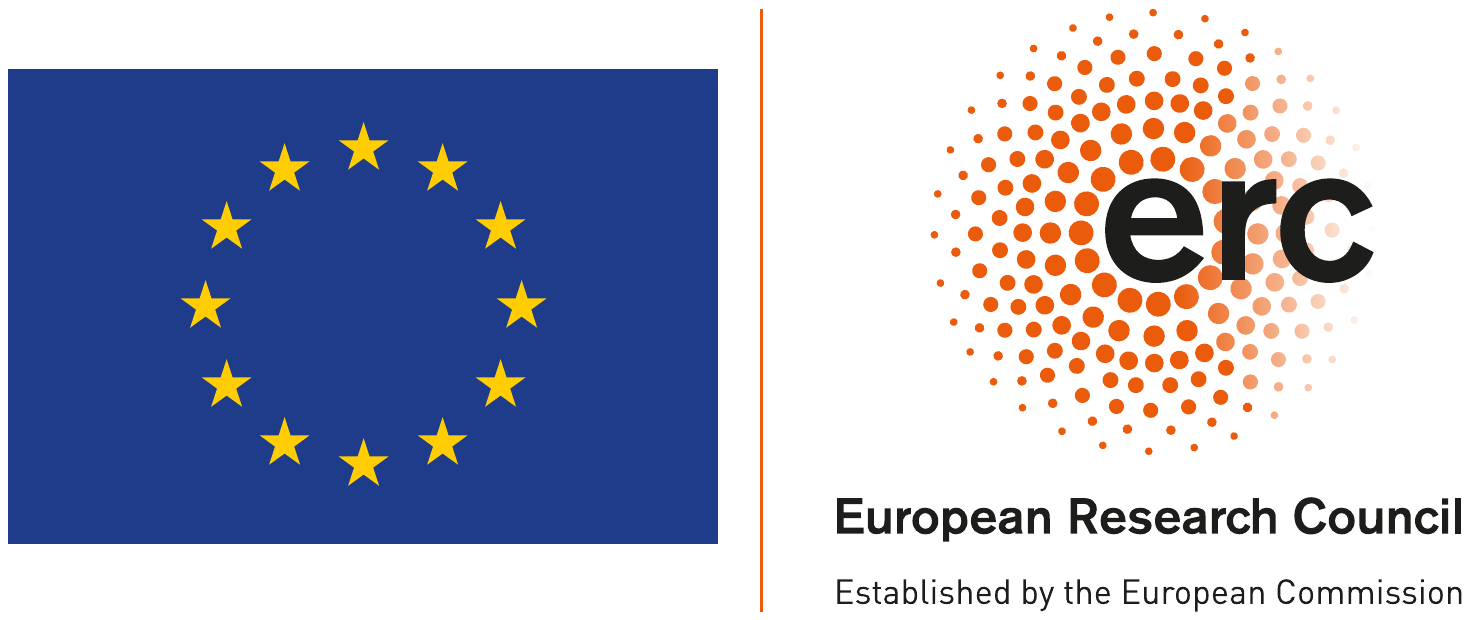}}
Funded by the European Union (\grantsponsor{501100000781}{ERC}{http://dx.doi.org/10.13039/501100000781}, FINABIS, \grantnum{501100000781}{101077902}).
Views and opinions expressed are however those of the authors only and do not necessarily reflect those of the European Union
or the European Research Council Executive Agency.
Neither the European Union nor the granting authority can be held responsible for them.
\end{acks}

\printbibliography
\end{document}